\documentclass[12pt]{article}
\usepackage{newpxtext}
\usepackage{amsthm,amsmath,amssymb}
\usepackage[normalem]{ulem}
\usepackage[bookmarks=true,hypertexnames=false,pagebackref]{hyperref}
\hypersetup{colorlinks=true, citecolor=blue, linkcolor=red, urlcolor=blue}
\usepackage{thm-restate}
\usepackage{cleveref}
\usepackage{fullpage}
\usepackage{bm}
\usepackage[ruled,linesnumbered]{algorithm2e}
\usepackage{algpseudocode}
\usepackage{enumerate}
\usepackage{xcolor}

\newtheorem{theorem}{Theorem}[section]

\newtheorem{fact}[theorem]{Fact}
\newtheorem{corollary}[theorem]{Corollary}
\newtheorem{lemma}[theorem]{Lemma}
\newtheorem{definition}[theorem]{Definition}

\newtheorem{observation}[theorem]{Observation}

\newcommand{\cH}{\mathcal{H}}
\newcommand{\tr}{\operatorname{Tr}}

\newcommand{\ket}[1]{\ensuremath{\left|#1\right\rangle}}
\newcommand{\op}[2]{|#1\rangle \langle #2|}

\def\eps{\epsilon}

\title{Learning marginals suffices!}

\begin{document}
\author{
Nengkun Yu
\thanks{Computer Science Department,
 Stony Brook University}
\and
Tzu-Chieh Wei
\thanks{CN Yang Institute for Theoretical Physics and Department of Physics and Astronomy,
 Stony Brook University}
}
\maketitle

\begin{abstract}
Beyond computer science, quantum complexity theory can potentially revolutionize multiple branches of physics, ranging from quantum many-body systems to quantum field theory. Our investigation centers on the intriguing relationship between the sample complexity of learning a quantum state and its circuit complexity. The circuit complexity of a quantum state represents the minimum quantum circuit depth needed to implement it. 
We derive a relationship between the circuit complexity of the quantum state and the unique determinism of the state by its marginals. As a direct consequence, the marginals uniquely determine the quantum state with low complexity. Additionally, our determination procedure proves robust, effectively handling potential noise in the marginals. Our result breaks the exponential barrier of the sample complexity for the tomography of quantum state with low complexity, even if only Pauli measurement is allowed. Our tool is quantum overlapping tomography, which only relies on random Pauli measurements.
Our proof overcomes difficulties characterizing short-range entanglement by bridging quantum circuit complexity and ground states of gapped local Hamiltonians. Our result, for example, settles the quantum circuit complexity of the multi-qubit GHZ state exactly.

\end{abstract}

\newpage
\section{Introduction}
\label{sec:intro}
The complexity of the physical systems provides an additional dimension to our understanding of the world~\cite{parisi1999complex,PhysRevLett.50.1946,PhysRevLett.43.1754}. The computational complexity theory offers a tool for studying the computational power of quantum computers and the resources required to solve computational problems. It also grants a new perspective on the nature of physical systems and how they function in quantum many-body physics. The circuit model offers a convenient way to quantify the complexity of a quantum state, which is the minimum depth of a quantum circuit that produces the state.

Most quantum states have complexities close to the maximum possible value through a counting argument~\cite{Nielsen:2011:QCQ:1972505}. However, despite being a long-standing open problem in quantum information theory~\cite{NLTS,https://doi.org/10.48550/arxiv.2206.13228}, it is highly challenging to establish a lower bound on quantum complexity for a given state because of the potential cancellation of gates~\cite{Haferkamp_2022} and quantum entanglement~\cite{RevModPhys.81.865}. 
Quantum entanglement is a phenomenon that describes the correlation between subsystems. Entanglement poses considerable challenges to the powerful tools of characterizing, simulating, and manipulating quantum systems~\cite{Shi_2006}, such as in quantum state tomography. This particular technique allows the reconstruction of the quantum state of a system from measurements and is one of the most indispensable tools for the development and verification of quantum technology. The sample complexity of quantum state tomography describes the protocol's efficiency, which refers to the number of measurements or samples required to reconstruct a quantum state using tomography techniques accurately~\cite{BBMR04,Keyl06,GJK08,FlammiaGrossLiuEtAl2012,KRT14,HHJ+16,https://doi.org/10.48550/arxiv.2206.05265,compressed,Flammia_2012,Gu__2020,doi:10.1137/1.9781611977554.ch47}. 
The sample complexity generically grows exponentially as the number of qubits increases~\cite{Holevo73,HHJ+16,OW16,OW17}. 

One can save many copies if we focus on learning or capturing only a part of that information instead of learning or representing the entire quantum state ~\cite{10.1145/3188745.3188802}. The recent active research of the {quantum overlapping tomography}~\cite{Cotler_2020} aims to output the classical representation of reduced density matrices rather than the entire quantum state. This surprising result~\cite{Cotler_2020} shows that all $k$-qubit reduced density matrices of an $n$ qubit state can be determined with at most $e^{O(k)} \log^2(n)$ rounds of Pauli measurements by leveraging the theory of perfect hash families. 
Huang, Kueng and Preskill designed classical shadow ~\cite{10.1145/3188745.3188802,Huang_2020,Huang_2022,evans2019scalable,10.1145/3406325.3451109} to predict $M$ functions of a quantum state using only a logarithmic number $O(\log(M))$ of measurements. This line of research opens the door to efficiently measuring many-body correlations and entanglement. 
In particular, \cite{Cotler_2020} leaves interesting questions of 
taking geometric
constraints into account and of 
exploring applications of quantum overlapping tomography in quantum state tomography. 

Despite the developments of the theoretical aspects of quantum computing, we are currently in the Noisy Intermediate-Scale Quantum (NISQ) era of quantum computing as regards  real quantum hardware. The available intermediate-scale quantum devices have limited coherence times, which makes it challenging to execute quantum circuits with considerable depth. Certifying the performance and accuracy of various near-term applications of quantum computers has become a critical challenge in quantum technologies. In the NISQ era, we are most likely facing quantum states as the output states of shallow quantum circuits. Intuitively, the description of a shallow quantum circuit only requires polynomial parameters. Therefore, it might be a manageable amount of information for a complete characterization.
Nevertheless, there is an exponential bottleneck in resource consumption in quantum state tomography. Will this obstacle prevent us from learning quantum systems in principle? The answer so far is indecisive, even for circuits with small depths. On the one hand, we can rigorously verify the intuition of efficient learning for a depth-$1$ quantum circuit because the output state is, at most, the tensor product of two-qubit states. On the other hand, the output state dramatically and quickly becomes unfathomable even for the quantum circuit of depth $4$~\cite{https://doi.org/10.48550/arxiv.quant-ph/0205133}.

This paper shows that \textit{exponential} number of samples is \textit{not necessary} for the tomography of quantum states with \textit{low circuit complexity}, even with only Pauli measurements. We prove that the output state of a depth-$D$ circuit, on general interaction graphs, is uniquely determined by its $2^D$-local reduced density matrices. Employing the geometrical locality can improve the bound to $2D$ on a 1-dimensional chain, and $\gamma_2(D)$~\footnote{We leave the definition of $\gamma_2(D)$ in Section 4. The upper bound listed is not tight.} on the square lattices with a combinatorial function $\gamma_2(D)\leq D^2+(D+1)^2$. We improved the original {quantum overlapping tomography} protocol to learn a set of reduced-density matrices. Our result is robust against perturbation in the following sense: any quantum state $\rho$, which has similar local reduced density matrices to those of a low-complexity quantum state $\ket{\psi}$, must be close to $\ket{\psi}$. In other words, we can treat the set of reduced density matrices as a robust fingerprint of quantum states with low complexity. Therefore, one can test whether a given quantum state is close to some low-complexity state or far from any low-complexity state with a small number of samples.
Our results rely on no assumption other than state complexity.

The intuition behind our findings originates from the observation that connects perhaps the two most important classes of quantum states: (1) the ground states of a local Hamiltonian and (2) the output states of quantum circuits. The output state of a shallow circuit is always the unique ground state of a local frustration-free Hamiltonian. Moreover, the spectral gap of this Hamiltonian has a constant lower bound.  We further observe that the set of unique ground states of local Hamiltonians are always uniquely determined by their local reduced density matrices among all mixed states (UDA).~\footnote{This is not true if we use "among all pure states" (UDP), instead of "among all mixed states."} This determination is robust against both statistical fluctuations from measurements and perturbation in the Hamiltonian as long as the gap is maintained throughout an entire short-range entangled phase of matter. 

As an application, our result also leads to a lower-bound method of the quantum state complexity. For example, the circuit complexity of 
$n$-qubit GHZ state is at least $\lceil \log n \rceil $ on general interaction graphs. This can be improved to $\lceil \frac{n}{2} \rceil $ on $1$-D chain, and $\max\limits_{D: \gamma_2(D)\leq n}\lceil D \rceil $ on square lattice. Interestingly, these are the exact complexity of generating $n$-qubit GHZ state. Our result also provides a lower-bound method for the complexity of the unitary through the correspondence between Choi states and unitaries, where the complexity of a unitary is the smallest number of the circuit depth among all the circuits implementing the unitary.

Our findings justify quantum state tomography from the viewpoint of reduced density matrices ~\cite{PhysRevLett.89.277906,PhysRevLett.89.207901,Xin_2017,Cramer_2010}, showing that it is precise and reliable against statistical fluctuations. It also confirms the intuition of \cite{Cotler_2020} that quantum overlapping tomography is a useful tool for the experimental characterization of many-body quantum states. 
More concretely, we refine the 
original protocol of quantum overlapping tomography by considering the geometric constraints of qubits. To develop this tool, we first study the state tomography using Pauli measurements \cite{Flammia_2012,Gu__2020} and improve the number of samples to ${\mathcal{O}}(\frac{10^n}{\epsilon^2})$. Then, we use this protocol to show ${\mathcal{O}}(\frac{10^k \log m}{\epsilon^2})$ samples are enough to output the classical representation of $m$ number of $k$-qubit reduced density matrices. We further show that this is tight for constant $k$.

\section{Preliminaries}

\subsection{Basic quantum mechanics}
An isolated physical system is associated with a
Hilbert space, called the {\it state space}. A {\it pure state} of a
quantum system is a normalized vector in its state space, denoted by the Dirac notation $\ket{\varphi}$. A
{\it mixed state} is represented by a density operator on the state
space. Here, a density operator $\rho$ on $d$-dimensional Hilbert space $\cH$ is a
semi-definite positive linear operator such that $\tr(\rho)=1$.

The state space of a composed quantum system is the tensor product of the state spaces of its component systems.
A Hilbert space $\bigotimes_{k=1}^{n}\cH_{k}$ is the tensor product of Hilbert spaces $\cH_{k}$.
The quantum state on the multipartite system $\bigotimes_{k=1}^{n}\cH_{k}$ is a semi-definite positive linear operator such that $\tr(\rho)=1$.

\subsection{Quantum measurements}

A positive-operator valued measure (POVM) on a finite-dimensional Hilbert space $\mathcal{H}$ is a set of positive semi-definite matrices $\mathcal{M}=\{M_i\}$ such that
\begin{align*}
\sum M_i=I_{\mathcal{H}}.
\end{align*}
Equivalently, a POVM $\mathcal{M}=\{M_i\}$ on $d$-dimensional $\mathcal{H}$ is informationally complete if the linear span of $\{M_i\}$ equals to the whole $d\times d$ matrix space. In a qubit system, it means $\sigma_I,\sigma_X,\sigma_Y$ and $\sigma_Z$ all live in the linear span of $\{M_i\}$, where $\sigma_I$, $\sigma_X,\sigma_Y$ and $\sigma_Z$ are the identity and the three Pauli matrices, respectively,
\begin{align*}
\sigma_I=\begin{bmatrix}1 &0\\0&1\end{bmatrix}, \sigma_X=\begin{bmatrix}0 &1\\1&0\end{bmatrix},  \sigma_Z=\begin{bmatrix}1 &0\\0&-1\end{bmatrix}, \sigma_Y=\begin{bmatrix}0 &i\\-i&0\end{bmatrix}.
\end{align*}
For example, the measurement described by $\mathcal{M}=\frac{1}{6}\{\sigma_I+\sigma_X,\sigma_I-\sigma_X,\sigma_I+\sigma_Y,\sigma_I-\sigma_Y,\sigma_I+\sigma_Z,\sigma_I-\sigma_Z\}$ is directly observed as an informationally complete measurement.
\begin{observation}
Given sufficient measurement outcomes of an informationally complete POVM, one can determine the state with high accuracy and confidence.
\end{observation}

\begin{observation}
For informationally complete POVMs, $\mathcal{M}_1$ on $\mathcal{H}_1$ and $\mathcal{M}_2$ on $\mathcal{H}_2$, $\mathcal{M}_1\otimes\mathcal{M}_2$ is an informationally complete POVM on $\mathcal{H}_1\otimes\mathcal{H}_2$.
\end{observation}
Specifically, $\mathcal{M}^{\otimes n}=\mathcal{M}=\frac{1}{6^n}\{\sigma_I+\sigma_X,\sigma_I-\sigma_X,\sigma_I+\sigma_Y,\sigma_I-\sigma_Y,\sigma_I+\sigma_Z,\sigma_I-\sigma_Z\}^{\otimes n}$ is an informationally complete POVM of an $n$-qubit system.
\begin{definition}
Let $X_1,X_2,\cdots,X_n$ be $n$ samples of a distribution on $\{1,2,\cdots,n\}$. Then the empirical distribution is defined as
\begin{align*}
\hat{p}(i)=\frac{\mathrm{number~of~occurrences~of}~i}{n}.
\end{align*}

\end{definition}

The following McDiarmid’s inequality \cite{mcdiarmid_1989} will be used in this paper.
\begin{lemma}\label{mc}
Consider independent random variables ${\displaystyle X_{1},X_{2},\dots X_{n}}$ on probability space $ {\displaystyle (\Omega ,{\mathcal {F}},{\text{P}})}$, where ${\displaystyle X_{i}\in {\mathcal {X}}_{i}}$ for all ${\displaystyle i}$, and there is a mapping ${\displaystyle f:{\mathcal {X}}_{1}\times {\mathcal {X}}_{2}\times \cdots \times {\mathcal {X}}_{n}\rightarrow \mathbb {R} }$. Assume there exist constants $ {\displaystyle c_{1},c_{2},\dots ,c_{n}} $ such that for all $ {\displaystyle i}$,
\begin{align}{\displaystyle {\underset {x_{1},\cdots ,x_{i-1},x_{i},x_{i}',x_{i+1},\cdots ,x_{n}}{\sup }}|f(x_{1},\dots ,x_{i-1},x_{i},x_{i+1},\cdots ,x_{n})-f(x_{1},\dots ,x_{i-1},x_{i}',x_{i+1},\cdots ,x_{n})|\leq c_{i}.} 
\end{align}
(In other words, changing the value of the ${\displaystyle i}$-th coordinate ${\displaystyle x_{i}}$ changes the value of ${\displaystyle f}$ by at most ${\displaystyle c_{i}}$.) Then, for any ${\displaystyle \epsilon >0}$,
\begin{align} 
{\displaystyle {\mathrm{Pr}}(f(X_{1},X_{2},\cdots ,X_{n})-\mathbb {E} [f(X_{1},X_{2},\cdots ,X_{n})]\geq \epsilon )\leq \exp \left(-{\frac {2\epsilon ^{2}}{\sum _{i=1}^{n}c_{i}^{2}}}\right)} .
\end{align}
\end{lemma}
\subsection{Local Hamiltonian}

A Hamiltonian is an operator representing a quantum system's total energy in quantum mechanics. A local Hamiltonian refers to a Hamiltonian operator expressed as a sum of terms, where each term involves only a few qubits in the system. More precisely, a $k$-local Hamiltonian of $n$ qubit system is of form
\begin{align*}
H=\sum_s H_s\otimes I_{\bar{s}},
\end{align*}
where $s$ may range over all subset of $\{0,\cdots,n-1\}$ with $|s|\leq k$, and $H_s$ with the property that $0\leq H_s \leq I_s$ is an operator acting on the sub-system $s$ (whose complement is denoted as $\bar{s}$). 
We call the set of $s$'s as the interaction graph.
For simplicity, we will also use 
\begin{align*}
H=\sum_s H_s,
\end{align*}
without explicitly writing out $I_{\bar{s}}$ in each term.

The ground energy of $H$ is the smallest eigenvalue of $H$. Ground space is the linear subspace spanned by the eigenvectors corresponding to the smallest eigenvalue. We call the normalized eigenvector a unique ground state if the ground space has only one dimension. The spectrum gap of Hamiltonian $H$ is the difference between the smallest eigenvalue and the second smallest eigenvalue.

A local Hamiltonian is called frustration-free if its ground energy is 0 (note that each term in the Hamiltonian here is positive semi-definite). The ground space of a local Hamiltonian $H=\sum_s H_s$ is the intersection of the null space of each local term $H_s$.

\subsection{Quantum Circuits}

The most popular and well-established model of quantum computing is the quantum circuit model, where quantum algorithms consist of a sequence of quantum gates applied to a set of qubits. In the \textit{quantum circuit model}, a quantum program consists of an instruction sequence $U_1 \cdots U_k$ and operates on an $n$-qubit register. For $1\leq i \leq k$, $U_i=\otimes_j V_{i,j}$ is the tensor product of two-qubit unitaries $V_{i,j}$ such that $V_{i,j_1}$ and $V_{i,j_2}$ apply on different qubits when $j_1\neq j_2$.

We take the initial state as $\ket{0}^{\otimes n}$. The meaning of the circuit is the matrix product $U_{k} \cdots U_{1}$ applied on the initial state  $\ket{0}^{\otimes n}$. We can regard $U_{i}$ as
a unitary matrix that applies to the entire $n$-qubit register. It is difficult to determine a quantum program's behavior using classical methods due to the state space's exponential size. This fact is known as the ``exponential wall" in quantum computing and is one of the main challenges in quantum program verification.

A quantum circuit depth refers to the number of sequential layers of quantum gates required to implement a quantum algorithm or operation on a quantum computer. Each layer consists of a set of quantum gates acting in different qubits. The depth of a quantum circuit is an essential metric for measuring the efficiency of a quantum algorithm, as it can impact the time required to execute the entire algorithm on a quantum computer.

\subsection{Uniquely determined (UD) by reduced density matrices}
Quantum state tomography via reduced
density matrices is an especially
promising resource-saving approach. It requires
the global state
has to be the only state compatible with its reduced density matrices. In other words,
it must be uniquely determined (UD) by its reduced density matrices \cite{PhysRevLett.89.277906,PhysRevLett.89.207901,Xin_2017}.

The UD criterion can be further classified into two categories: uniquely determined among all states (UDA) and
uniquely determined among pure states by local
reduced density matrices (UDP)~\cite{PhysRevA.88.012109}.

A quantum state $\psi=\op{\psi}{\psi}$ is UDP by its reduced density matrices $\psi_{s_i}$ (supported on a set $s_i$ of sites, each with a finite Hilbert space) for $1\leq i\leq m$ if
\begin{align*}
\phi=\op{\phi}{\phi},\ \&  \ \phi_{s_i}=\psi_{s_i}, \forall 1\leq i\leq m\  \Longrightarrow \  \phi=\psi,
\end{align*}
where we use the same symbol $\phi$ but with subscript $s_i$ to denote the reduced density matrix obtained by tracing out the complement of $s_i$ from the state $\phi$; and similarly for $\psi_{s_i}$.

A quantum state $\psi=\op{\psi}{\psi}$ is UDA by its reduced density matrices $\psi_{s_i}$ for $1\leq i\leq m$ if for some generally mixed state $\rho$,
\begin{align*}
 \rho_{s_i}=\psi_{s_i}, \forall 1\leq i\leq m\  \Longrightarrow \  \rho=\psi.
\end{align*}

In other words, UDP means no other pure state shares the same set of reduced density matrices;  UDA means no other state shares the same set of reduced density matrices.
It is important to note that UDP does not imply UDA~\cite{Xin_2017}.

\section{Quantum Overlapping Tomography from Random Pauli Measurements}

In this section, We start with an observation about Pauli measurements, common knowledge for quantum information experimentalists. Then, we employ this observation to improve previous results about quantum state tomography using Pauli measurements.

\subsection{ An Observation about Pauli Measurements}
When we measure an element of the Pauli group, for instance, $\sigma_X\otimes \sigma_Y$, on a two-qubit state $\rho$, the outcome is a sample from a $4$-dimensional probability distribution, says $(p_{00},p_{01},p_{10},p_{11})$, such that
\begin{align*}
\tr(\rho(\sigma_X\otimes \sigma_Y))=p_{00}-p_{01}-p_{10}+p_{11}.
\end{align*}

One can easily observe that
\begin{align*}
\tr(\rho(\sigma_X\otimes \sigma_I))=p_{00}+p_{01}-p_{10}-p_{11},\\
\tr(\rho(\sigma_I\otimes \sigma_Y))=p_{00}-p_{01}+p_{10}-p_{11},\\
\tr(\rho(\sigma_I\otimes \sigma_I))=p_{00}+p_{01}+p_{10}+p_{11}.
\end{align*}

In other words, by measuring $XY$, we obtained a sample of $\sigma_X\sigma_I$, a sample of $\sigma_I\sigma_Y$, and a sample of $\sigma_I\sigma_I$. 

For a general $n$-qubit system, we have the following observation.
\begin{observation}
For any $P=P_1\otimes P_2\otimes\cdots\otimes P_n\in\{\sigma_X,\sigma_Y,\sigma_Z\}^{\otimes n}$, the measurement result of performing measurement $P_i$ on the $i$-th qubit is an $n$-bit string $s$. One can interpret the measurement outcome of performing $Q_i\in\{\sigma_I,\sigma_X,\sigma_Y,\sigma_Z\}$ on the $i$-th qubit if $Q_i=P_i$ or $Q_i=\sigma_I$. We call those $Q=Q_1\otimes Q_2\otimes\cdots\otimes Q_n$'s correspond to $P$.
\end{observation}

\subsection{State Tomography Using Pauli Measurements}
Our measurement scheme follows: For any $\epsilon>0$, fix an integer $m$.

1. For any $P\in\{\sigma_X,\sigma_Y,\sigma_Z\}^{\otimes n}$, one performs $m$ times $P$ on $\rho$, and records the $m$ samples of the $2^n$-dimensional outcome distribution.

According to the key observation, this measurement scheme provides $m\cdot 3^{n-w}$ samples of the expectation $\tr(\rho P)$, say, $\frac{\mu_P}{m\cdot 3^{n-w}}$, for each Pauli operator $P\in \{\sigma_I,\sigma_X,\sigma_Y,\sigma_Z\}^{\otimes n}$ with weight $w$, where $-m\cdot 3^{n-w}\leq\mu_P\leq m\cdot 3^{n-w}$.

2. Output 
\begin{align*}
\sigma=\sum_P \frac{\mu_P}{m\cdot 3^{n-w}\cdot 2^n} P.
\end{align*}

Using this scheme, we obtained $m\cdot 3^n$ independent samples,
\begin{align*}
X_1,X_2,\cdots, X_{m\cdot 3^n},
\end{align*}
each  $X_i$ is an $n$-bit string recording outcomes on all qubits (using bit 0 to denote the +1 eigenvalue and bit 1 to denote the -1 eigenvalue of the measured Pauli operator).   
Given that each operator is measured $m$ times, specifically, we assign that $X_1,X_2,\cdots,X_{m}$ correspond to the measurement $\sigma_X^{\otimes n}$, $X_{m+1},X_{m+2}$, and $\cdots,X_{2m}$ corresponds to the measurement $\sigma_X^{\otimes n-1}\otimes \sigma_Y$, \dots, and until $\sigma_Z^{\otimes n}$.

We observe that for any $P$ of weight $w$,
$\mu_P=\sum_{j=0}^{m\cdot 3^{n-w}-1} Z_j$,
where $Z_j$ are independent samples from the distribution $Z$
\begin{align*}
\mathrm{Pr}(Z=1)=\frac{1+\tr (\rho P)}{2},\\
\mathrm{Pr}(Z=-1)=\frac{1-\tr (\rho P)}{2}.
\end{align*}
Clearly,  we have
\begin{align*}
\mathbb{E}(Z)&=\tr (\rho P), \mathbb{E}(Z^2)=1,\\
\mathbb{E}(\mu_P)&=m\cdot 3^{n-w}\cdot\tr (\rho P),\\
\mathbb{E}(\mu_P^2)&=\mathbb{E}(\mu_P)^2+var(\mu_P)=\mathbb{E}(\mu_P)^2+m\cdot 3^{n-w}[var(Z)]=m^2\cdot 9^{n-w}\cdot \tr^2 (\rho P)+m\cdot 3^{n-w}(1-\tr^2 (\rho P)).
\end{align*}
Furthermore, $Z_j$ can be obtained from samples
\begin{align*}
X_1,X_2,\cdots, X_{m\cdot 3^n}.
\end{align*}

Therefore,
\begin{align*}
\sigma=\sum_P \frac{\mu_P}{m\cdot 3^{n-w_P}\cdot 2^n} P.
\end{align*}
is defined according to the samples
\begin{align*}
X_1,X_2,\cdots, X_{m\cdot 3^n}.
\end{align*}
We verify that
\begin{align*}
\mathbb{E}\sigma=\rho,
\end{align*}
where the expectation is taken over the probabilistic distribution according to the measurements.

For any 
\begin{align*}
\rho=\sum_P \frac{\alpha_P}{2^n} P,
\end{align*}
we can define the function $f:X_1\times X_2\times \cdots\times X_{m\cdot 3^n}\mapsto \mathbb{R}$
\begin{align*}
f=||\rho-\sigma||_2=\sqrt{{\rm Tr}[(\rho-\sigma)^\dagger (\rho-\sigma)]}.
\end{align*}
According to Cauchy's inequality, we have
\begin{align*}
&\mathbb{E}f \\
\leq& \sqrt{\mathbb{E}f^2}\\
=& \sqrt{\mathbb{E} (\tr\rho^2-2\tr \rho\sigma+\tr\sigma^2)}\\
=& \sqrt{\mathbb{E}\tr\sigma^2-\tr\rho^2}\\
=& \sqrt{\frac{1}{2^n}\sum_P (\mathbb{E} \frac{\mu_P^2}{m^2\cdot 9^{n-w_P}}-\alpha_P^2)} \\
=& \sqrt{\frac{1}{2^n}\sum_P (\frac{m^2\cdot 9^{n-w_P}\cdot \alpha_P^2+m\cdot 3^{n-w_P}(1-\alpha_P^2)}{m^2\cdot 9^{n-w_P}}-\alpha_P^2)} \\
=&\sqrt{\frac{1}{m\cdot 2^n}\cdot{\sum_P \frac{1-\alpha_P^2}{3^{n-w_P}}}}\\
<&\sqrt{\frac{1}{m\cdot 2^n}\cdot{\sum_P \frac{1}{3^{n-w_P}}}}\\
=& \sqrt{\frac{1}{m\cdot 2^n}\cdot{\sum_{w_P=0}^n \frac{1}{3^{n-w_P}}{{n}\choose{w_P}}3^{w_P}}}\\
=&\sqrt{\frac{1}{m\cdot 6^n}\cdot (1+9)^n} \\
=& \sqrt{\frac{5^n}{m\cdot 3^n}}.
\end{align*}

For any sample $X_i$ corresponding to $P\in\{\sigma_X,\sigma_Y,\sigma_Z\}^{\otimes n}$, if only $X_i$ is changed, $\mu_Q$ would be changed only for those $Q\in\{\sigma_I,\sigma_X,\sigma_Y,\sigma_Z\}^{\otimes n}$ where
$Q$ is obtained by replacing some $\{\sigma_X,\sigma_Y,\sigma_Z\}$'s of $P$ by $\sigma_I$.
Moreover, the resultant value of $\mu_Q$ would change by two at most.
According to the triangle inequality, $f$ would change at most 
\begin{align*}
&\Big\Vert\sum_Q  \frac{\Delta\mu_Q}{m\cdot 3^{n-w_Q}\cdot 2^n} Q\Big\Vert_2\\
=&\sqrt{\sum_Q \frac{\Delta\mu_Q^2}{m^2\cdot 9^{n-w_Q}\cdot 2^n}}\\
\leq& \sqrt{\sum_Q \frac{2^2}{m^2\cdot 9^{n-w_Q}\cdot 2^n}}\\
=& \sqrt{\sum_{w_Q=0}^n \frac{2^2}{m^2\cdot 9^{n-w_Q}\cdot 2^n}{{n}\choose{w_Q}}}\\
=&\frac{2\cdot \sqrt{5}^n}{m\cdot 3^n},
\end{align*}
where $Q$ ranges over all Paulis which correspond to $P$'s, and $\Delta\mu_Q$ denotes the difference of $\mu_Q$ when $X_i$ is changed.

We only consider $\delta<1/3$, then $\log(1/\delta)>1$.
For any $\epsilon'>0$, by choosing $m=(3+2\sqrt{2})\cdot\frac{5^n\log\frac{1}{\delta}}{3^n\cdot \epsilon'^2}$, we have $\mathbb{E}f< (\sqrt{2}-1){\epsilon'}$.
Therefore,
\begin{align*}
&\mathrm{Pr}(f> \epsilon')\\
<&\mathrm{Pr}(f-\mathbb{E}f>(2-\sqrt{2}){\epsilon'}) \\
<&\exp(-\frac{(12-8\sqrt{2})\cdot \epsilon'^2}{4\cdot \frac{5^n}{m^2\cdot 9^n}\cdot m\cdot 3^n}) \\
=&\exp(-\frac{m\cdot(3-2\sqrt{2})\cdot 3^n\cdot\epsilon'^2}{5^n}) \\
<&\delta, 
\end{align*}
where the inequality is by Lemma \ref{mc} (McDiarmid’s inequality).

For a general quantum state and $\epsilon>0$, we let $\epsilon'=\frac{\epsilon}{\sqrt{2^n}}$, and know that
$||\rho-\sigma||_1>\epsilon$ implies $||\rho-\sigma||_2>\epsilon'$. Therefore,
\begin{align*}
&\mathrm{Pr}(||\rho-\sigma||_1>\epsilon)\\
\leq& \mathrm{Pr}(||\rho-\sigma||_2>\epsilon')\\
=&\mathrm{Pr}(f>\epsilon').
\end{align*}

The total number of used copies is 
\begin{align*}
m\cdot 3^n=(3+2\sqrt{2})\cdot\frac{10^n\log\frac{1}{\delta}}{\epsilon^2}.
\end{align*}

If we know the rank of $\rho$ is at most $r$, first we will find a $\sigma'$ with rank at most $r$, which minimizes
\begin{align*}
||\sigma-\sigma'||_2
\end{align*}
Accordingly, this implies that \begin{align*}
||\sigma-\sigma'||_2\leq ||\rho-\sigma||_2, 
\end{align*}
and we have
\begin{align*}
||\rho-\sigma'||_2\leq ||\sigma-\sigma'||_2+||\rho-\sigma||_2\leq 2||\rho-\sigma||_2. 
\end{align*}
By choosing $\epsilon'=\frac{\epsilon}{\sqrt{2r}}$, we know that
$||\rho-\sigma||_1>\epsilon$ implies $||\rho-\sigma||_2>\epsilon'$. We only need to choose $m=4(3+2\sqrt{2})\cdot\frac{5^n\log\frac{1}{\delta}}{3^n\cdot \epsilon'^2}$ to make $\mathbb{E}f<\frac{(\sqrt{2}-1)\epsilon'}{2}$,
Therefore,
\begin{align*}
&\mathrm{Pr}(||\rho-\sigma'||_1>\epsilon)\\
\leq&\mathrm{Pr}(||\rho-\sigma'||_2>\epsilon')\\
\leq&\mathrm{Pr}(||\rho-\sigma||_2>\frac{\epsilon'}{2})\\
\leq & \mathrm{Pr}(f-\mathbb{E}f>\frac{(2-\sqrt{2})\epsilon'}{2})\\
=&\exp(-\frac{m\cdot 3^n\cdot\epsilon'^2}{4(3+2\sqrt{2})\cdot 5^n})\\
<&\delta
\end{align*}
The total number of used copies is thus
\begin{align*}
m\cdot 3^n=8(3+2\sqrt{2})\cdot\frac{5^n\cdot r\cdot\log\frac{1}{\delta}}{\epsilon^2}.
\end{align*}
\subsection{Joint Measurement Lower bound}

In this subsection, we show that $\Omega(\frac{\log(n/\delta)}{\epsilon^2})$ copies are necessary for quantum overlapping tomography by proving the following: $\Omega(\frac{\log(n/\delta)}{\epsilon^2})$ copies are necessary for quantum overlapping tomography with $k=1$, even if general \textit{joint measurement} is used.

Since trace distance is non-increasing under the action of partial trace, we can conclude that any measurement scheme that can solve the quantum overlapping tomography problem for $k\geq 1$ automatically solves the case that $k=1$. Moreover, we focus on classical distributions to deal with general joint measurement schemes.

We first consider the following simple question. Given a binary random variable $X$ that obeys  either distribution $q_0=(1/2-\epsilon,1/2+\epsilon)$ or $q_1=(1/2+\epsilon,1/2-\epsilon)$, which distribution is the true distribution? For any fixed $m$, the number of tossing this coin, the best strategy is to toss the coin $m$ times and declare the index ($0$ or $1$) that appears less.

Let the $X_1,X_2,\dots,X_m$ be the $m$ samples of $q_1$ and any $0\leq t\leq 2m\epsilon$, we then employ the result in~\cite{Mousavi2016} and obtain:
\begin{align}
\mathrm{Pr}\Big(\sum_{i=1}^m X_i>t+m(1/2-\epsilon)\Big)\geq \frac{1}{4}\cdot \exp\Big(-\frac{2t^2}{m(1/2-\epsilon)}\Big).
\end{align}
By choosing $t=m\epsilon$, $\mathrm{Pr}\Big(\sum_{i=1}^m X_i>t+m(1/2-\epsilon)\Big)=\mathrm{Pr}\big(\sum_{i=1}^m X_i>m/2\big)$ is the probability of answering $q_0$, a lower bound on the failure probability.

To succeed with probability at least $1-\delta'$, we must have
\begin{align}
\delta'\geq\frac{1}{4}\cdot \exp\Big(-\frac{2m\epsilon^2}{1/2-\epsilon}\Big).
\end{align}
That is, $\frac{1-2\epsilon}{4\epsilon^2}\log(\frac{1}{4\delta'})$ samples are needed to distinguish $q_0$ and $q_1$ with confidence at least $1-\delta'$. 

By referring back to our problem of showing $\Omega\big(\frac{\log(n/\delta)}{\epsilon^2}\big)$ samples of $n$-qubit state $\rho$ are necessary to solve the quantum overlapping tomography for $k\geq 1$, to within additive error $\epsilon$ and confidence at least $1-\delta$, we consider the classical distributions $p_{i_1,i_2,\cdots,i_n}=q_{i_1}\otimes q_{i_2}\otimes\cdots \otimes q_{i_n}$, where each $i$ is either 0 or 1 and $q_0=(1/2-\epsilon,1/2+\epsilon)$ and $q_1=(1/2+\epsilon,1/2-\epsilon)$. In total, there are $2^n$ different distributions.

Suppose a quantum procedure $\mathcal{A}$ uses $m$ copies of $\rho$ to accomplish the quantum overlapping tomography with at least $1-\delta$ probability. 

Let $Z_1,Z_2,\cdots,Z_n$ be random variables that obey the uniform binary distribution.
Choose each $p_{Z_1,Z_2,\cdots,Z_n}$ with probability $1/2^n$, and apply $\mathcal{A}$ on $m$ copies of $p_{Z_1,Z_2,\cdots,Z_n}$. Because the $\ell_1$ norm is non-increasing under partial trace, we know that according to the output of $\mathcal{A}$, we can successfully recover the indices $Z_1,Z_2,\cdots,Z_n$ with probability at least $1-\delta$.

In the following, we first observe that any quantum procedure does not help in recovering $Z_1,Z_2,\cdots,Z_n$ from samples of $p_{Z_1,Z_2,\cdots,Z_n}$.
We assume the joint measurement $(M_{0,0,\cdots,0},\cdots,M_{1,1,\cdots,1})$ applied on $m$ copies (samples) of $p$ such that the measurement outcome $M_{i_1,i_2,\cdots,i_n}$ allows us to answer $Z_1=i_1,Z_2=i_2,\cdots,Z_n=i_n$. Here $M_{i_1,i_2,\cdots,i_n}$'s are $2^{mn}\times 2^{mn}$ matrices.

First, we observe that $p_{Z_1,Z_2,\cdots,Z_n}$'s are all diagonal, so are $p_{Z_1,Z_2,\cdots,Z_n}^{\otimes m}$.
Hence, the off-diagonal elements of $M_{i_1,i_2,\cdots,i_n}$ does not affect this task.
Therefore, we only need to consider the procedure in the following two steps: The first step measures $m$ copies of $p_{Z_1,Z_2,\cdots,Z_n}$'s in the diagonal basis, and the second step outputs according to certain probability distributions.

The first step ensures that we only measure each copy of $p_{Z_1,Z_2,\cdots,Z_n}$'s on a diagonal basis since there is no difference. According to the convexity of the successful probability, we know that the deterministic function works best in the second step, declaring the index ($0$ or $1$) that appears less for each $1\leq j\leq n$.

Now, we assume the output random variable is $Y_1,Y_2,\cdots,Y_n$, our goal is
\begin{align}
\mathrm{Pr}(Y_1=Z_1,Y_2=Z_2,\cdots,Y_n=Z_n)\geq 1-\delta.
\end{align}
By Bayes' theorem, we know that
\begin{align*}
&\mathrm{Pr}(Y_1=Z_1,Y_2=Z_2,\cdots,Y_n=Z_n)\\
=&\mathrm{Pr}(Y_1=Z_1)\times \mathrm{Pr}(Y_2=Z_2|Y_1=Z_1)\times\cdots\times \mathrm{Pr}(Y_n=Z_n|Y_1=Z_1,\cdots,Y_{n-1}=Z_{n-1})\\
\leq& (1-\delta')\cdot (1-\delta')\cdot\cdots\cdot(1-\delta')\\
=& (1-\delta')^n,
\end{align*}
where we use the fact that $p_{Z_1,Z_2,\cdots,Z_n}$'s are all in a tensor product form, and therefore, $\mathrm{Pr}(Y_2=Z_2|Y_1=Z_1)=\mathrm{Pr}(Y_2=Z_2)$, and so on.
We note that we have also denoted by $(1-\delta')$ the successful probability of discriminating $q_0$ and $q_1$ with $m$ copies. 

Therefore, we require that 
\begin{align*}
(1-\delta')^n>1-\delta.
\end{align*}
That is $\delta'=\Theta(\frac{\delta}{n})$. It implies the bound $m=\Omega(\frac{\log(n/\delta)}{\epsilon^2})$.

\subsection{Pauli Measurement Upper bound}
In this subsection, we analyze Algorithm 1:
\begin{algorithm}
\caption{Quantum Overlapping Tomography by Pauli Measurements}
\begin{algorithmic}
\item Repeat the following measurement $32\cdot 10^k\cdot\epsilon^{-2}\cdot\log\big(2{{n}\choose{k}}/\delta\big)$ times\;
\item For $i=1$ to $n$:
measure the $i$-th qubit in a randomly chosen basis from $\{\sigma_X,\sigma_Y,\sigma_Z\}$\;
\end{algorithmic}
\end{algorithm}

By estimating the binomial coefficient, we observe the following, 
\begin{observation}
In Algorithm 1, for any $S=\{i_1,i_2,\cdots,i_k\}\subseteq\{1,2,\cdots,n\}$, 
with probability at least $1-\frac{2^m}{e^m}\cdot 3^k$, each $P\in \{\sigma_X,\sigma_Y,\sigma_Z\}^{\otimes k}$ was measured for estimating $\rho_S$ at least $m=(3+2\sqrt{2})\cdot \frac{10^k\cdot\log(2{{n}\choose{k}}/\delta)}{3^k\cdot\epsilon^{2}}$ times.
\end{observation}

According to Theorem 1, the tomography of $\rho_S$ with trace distance error $\epsilon$ was successful with probability at least 
\begin{align*}
1- \Big[3^k\cdot \frac{2^m}{e^m}+\frac{\delta}{2{{n}\choose{k}}}\Big]> 1-\frac{\delta}{{{n}\choose{k}}},
\end{align*}
where the last inequality follows from
\begin{align*}
3^k\cdot \frac{2^m}{e^m}<  3^k \Big(\frac{2}{e}\Big)^{12\cdot 3^k}\cdot \Big(\frac{2}{e}\Big)^{4\cdot 3^k\cdot \log\big(2{{n}\choose{k}}/\delta\big)}<  3^k \Big(\frac{2}{e}\Big)^{12\cdot k}\cdot \Big(\frac{1}{e}\Big)^{ \log\big(2{{n}\choose{k}}/\delta\big)}
< \frac{\delta}{2{{n}\choose{k}}}.
\end{align*}
By the union bound, the quantum overlapping tomography with trace distance error $\epsilon$ was successful with a probability of at least 
$1- {{n}\choose{k}}\cdot \frac{\delta}{{{n}\choose{k}}}=1-\delta$.


\section{A property of the unique ground state}

In this section, we illustrate the following relationship between the unique ground state of the local Hamiltonian and UDA.

Let $G=\{s_1,\cdots,s_m\}$ be the interaction graph.
Let us look at the set of all $n$-qubit mixed states 
\begin{align*}
\mathcal{D}=\{\rho| \rho^\dagger=\rho,\,\rho\geq 0,\ \tr \rho=1\},
\end{align*}
and a linear map $L$ which maps the set $\mathcal{D}$ into the set of tuples of reduced density matrices
\begin{align*}
\mathcal{R}=\{(\sigma_{s_1},\cdots,\sigma_{m})| \forall \sigma\in \mathcal{D}\}.
\end{align*}
We note that $\mathcal{R}$ is convex.  
\begin{fact}\label{Hamiltonian}
(i) A quantum state $\ket{\psi}$ is UDA by its $k$-local reduced density matrices on the interaction graph $G$ only if the tuple of its reduced density matrices is an extreme point of $\mathcal{R}$.

(ii) A quantum state $\ket{\phi}$ is the unique ground state of a Hamiltonian with the interaction graph $G$ only if the tuple of its reduced density matrices is an exposed point \footnote{An exposed point of a convex set is a point at which some continuous linear functional attains its strict maximum (or equivalently minimum in the case of Hamiltonians) over the convex set.} of $\mathcal{R}$.

(iii) A quantum state $\ket{\phi}$ is the unique ground state of a Hamiltonian with the interaction graph $G$ only if it is UDA by its reduced density matrices on the interaction graph $G$. Conversely, if a pure quantum state $|\phi\rangle$ is UDA by its reduced density matrices on some interaction graph $G$ and the tuple of these matrices is also an exposed point, then 
$|\phi\rangle$ is a unique ground state of some local Hamiltonian with the same interaction graph $G$.

(iv) The exposed point does not always imply UDA, even UDP.
\end{fact}
\begin{proof}
(i) Suppose $\ket{\psi}$ is UDA by its $k$-local reduced density matrices, we will prove $(\psi_{s_1},\cdots,\psi_{m})$
is an extreme point of $\mathcal{R}$.

We denote the $k$-reduced density matrices of $\ket{\psi}$ as $\psi_{s_i}$ for $1\leq i\leq m$.

For any two points $(\sigma_{s_1},\cdots,\sigma_{s_{m}})$ and $ (\tau_{s_1},\cdots,\tau_{s_{m}})\in \mathcal{R}$ such that
\begin{align*}
\frac{(\sigma_{s_1},\cdots,\sigma_{s_{m}}) + (\tau_{s_1},\cdots,\tau_{s_{m}})}{2}=(\psi_{s_1},\cdots,\psi_{s_{m}}),
\end{align*}
we first examine pre-images of $(\sigma_{s_1},\cdots,\sigma_{s_{m}})$ and $(\tau_{s_1},\cdots,\tau_{s_{m}})$, denoted as $\sigma$ and $\tau$, respectively.
We know that $\frac{\sigma+\tau}{2}\in \mathcal{D}$, and it shares the same set of $k$-local reduced density matrices with $\psi$, as the map $L$ is linear. 
Because $\ket{\psi}$ is UDA by its $k$-local reduced density matrices, we thus infer that
\begin{align*}
\frac{\sigma+\tau}{2}=\op{\psi}{\psi}.
\end{align*}
Because $\op{\psi}{\psi}$ is a pure state and hence is an extreme point of $\mathcal{D}$, we conclude that
\begin{align*}
\sigma=\tau=\op{\psi}{\psi}.
\end{align*}
This leads to
\begin{align*}
(\sigma_{s_1},\cdots,\sigma_{s_{m}})=(\tau_{s_1},\cdots,\tau_{s_{m}})=(\psi_{s_1},\cdots,\psi_{s_{m}}).
\end{align*}
Therefore, $(\psi_{s_1},\cdots,\psi_{s_{m}})$ is an extreme point of the convex set $\mathcal{R}$.



(ii) If $\ket{\psi}$  is the unique ground state of a $k$-local Hamiltonian $H=\sum_{s_i} H_{s_i}$. We observe that
\begin{align*}
\tr(H\op{\psi}{\psi})=\sum_{s_i}\tr (H_{s_i}\psi_{s_i}). 
\end{align*}
This means we can regard $H$ as a linear functional on the space $\mathcal{R}$ and reaches its strict extremum (i.e., minimum here) at the $(\psi_{s_1},\cdots,\psi_{s_{m}})$. This means $(\psi_{s_1},\cdots,\psi_{s_{m}})$ is an exposed point of $\mathcal{R}$.

(iii) If $\ket{\psi}$ is the unique ground state of $H$, then
\begin{align*}
\tr(H \psi)=\sum_{s_i} \tr(H_{s_i} \psi_{s_i})
\end{align*}
is the ground-state energy.
For $\rho=\sum_{j} p_j \op{\phi_j}{\phi_j}$ such that
\begin{align*}
\rho_{s_i}=\psi_{s_i}, \ \ \ \forall 1\leq i\leq m,
\end{align*}
we have ${\rm Tr}(H\rho)={\rm Tr}(H\phi)$, i.e., $\rho$ has the same energy as the ground-state energy. Because of the uniqueness of the ground state, this further implies that $\rho=\psi$.
In other words, the ground state $\ket{\psi}$  is UDA by its local reduced density matrices on $G$.

Suppose $\ket{\phi}$ is UDA by its local reduced density matrices on $G$, which is also an exposed point. 
There exists a hyperplane that contains $(\phi_{s_1},\cdots,\phi_{s_{m}})$ while keeping all other points in $\mathcal{R}$ at one side on this plane. This hyperplane corresponds to a linear function $f:\mathcal{R}\mapsto \mathbb{R}$ such that
\begin{align*}
&f(\rho_{s_1},\cdots,\rho_{s_{m}})\geq 0, \ \forall (\rho_{s_1},\cdots,\rho_{s_{m}})\in\mathcal{R}\\
&f(\phi_{s_1},\cdots,\phi_{s_{m}})= 0.
\end{align*}
Any linear function on $\mathcal{R}$ is of form
\begin{align*}
f(\rho_{s_1},\cdots,\rho_{s_{m}})=\sum_{s_i} \tr(H_{s_i}\rho_{s_i}),
\end{align*}
for a set of Hermitian operators $\{H_{s_i}\}$. 
Hence,
$\ket{\psi}$ is a unique ground state of the $k$-local Hamiltonian $H=\sum_{s_i} H_{s_i}$. 

Any ground state $\rho$ of $H=\sum_{s_i} H_{s_i}$ must satisfies $\rho_{s_i}=\phi_i$.
Because it is UDA, we know that $\ket{\phi}$ is the unique ground state of $H$.

(iv) Consider the very simple case $n=2$ and $G=\{s\}$ with $s=\{1\}$.
Now $(\op{0}{0})$ is an exposed point of the $\mathcal{R}$, which is the set of all $1$-qubit states.
Any two-qubit product state $|0\rangle\otimes(\alpha|0\rangle+\beta|1\rangle)$ is compatible with this exposed point, and hence UDA nor UDP does not hold.

\end{proof}

We remark that the geometry of reduced-density matrices has been extensively studied in previous works. From the geometric picture,  the authors of~\cite{Verstraete_2006} observed that, for an interacting spin system, ``the most extreme points in the convex set of reduced density operators uniquely characterize a state.'' Reference~\cite{https://doi.org/10.48550/arxiv.1606.07422} studied the geometry of reduced density matrices by projecting the set to $\mathbb{R}^3$ from the joint numerical range point of view~\cite{Pucha_a_2011,Gawron_2010}. The relationship between UDA and a local Hamiltonian's unique ground state has been studied by exploring the extreme and exposed points of the convex body in \cite{Wyderka_2018,Karuvade_2019}.

Fact \ref{Hamiltonian} indicates that the tuple of reduced density matrices is a fingerprint of the unique ground state of the local Hamiltonian. Here, we do not explicitly construct $H$, which depends on the detailed structure of $\mathcal{R}$. Moreover, $H$ is generally not frustration-free. The following observation shows this fingerprint is robust for the gapped local Hamiltonian.

\begin{lemma}\label{lemma:4.2}\label{ground}
Let $\ket{\psi}$ be the unique ground state of a $k$-local Hamiltonian $H=\sum_i H_{s_i}$ with $0\leq H_{s_i}\leq I_{s_i}$, gap $\Delta>0$ and interaction graph $G=\{s_1,\cdots,s_m\}$, for any state $\rho$, one of the following conditions must be satisfied:
\begin{enumerate}
\item $||\psi-\rho||_1<\eps$;
\item  $||\psi_{s_i}-\rho_{s_i}||_1>\frac{\Delta \eps^2}{4 {m}}$ for some $i$,
\end{enumerate}
where $\psi\equiv\op{\psi}{\psi}$ denotes the density matrix of state $|\psi\rangle$ and $\psi_{s_i}$ denotes the corresponding reduced density matrix supported on the subsystem $s_i$.
\end{lemma}

\begin{proof} Let the ground energy of $H$ be $\lambda$, then
we have
\begin{align*}
    H\geq \lambda \op{\psi}{\psi}+(\lambda+\Delta)(I-\op{\psi}{\psi})=(\lambda+\Delta)I-\Delta \op{\psi}{\psi}.
\end{align*}
Let us assume that scenario one above does not hold. Then, by the relation between quantum fidelity and trace norm~\cite{Nielsen:2011:QCQ:1972505}, we know that 
\begin{align*}
||\psi-\rho||_1\ge\eps \Longrightarrow \tr (\op{\psi}{\psi}\rho)\leq 1- \frac{||\psi-\rho||_1^2}{2^2}\leq 1-\frac{\eps^2}{4}.
\end{align*}
We can obtain the following bound
\begin{align*}
\tr(H\rho)\geq \tr [((\lambda+\Delta)I-\Delta \op{\psi}{\psi})\rho]=(\lambda+\Delta)-\Delta \tr (\op{\psi}{\psi}\rho).
\end{align*}
Then 
\begin{align*}
\tr[H(\rho-\psi)]\geq (\lambda+\Delta)-\Delta \tr (\op{\psi}{\psi}\rho)-\lambda=\Delta(1-\tr (\op{\psi}{\psi}\rho))\geq \frac{\Delta \eps^2}{4}.
\end{align*}
Therefore,
\begin{align*}
\tr[H(\rho-\psi)]=\sum_{s_i} \tr[H_{s_i}(\rho_{s_i}-\psi_{s_i})]\geq \frac{\Delta \eps^2}{4}.
\end{align*}
This means that there must exist some $s_i$ such that
\begin{align*}
\tr[H_{s_i}(\rho_{s_i}-\psi_{s_i})]\geq \frac{\Delta \eps^2}{4 m}.
\end{align*}
Since $0\leq H_{s_i}\leq I_{s_i}$, we have
\begin{align*}
||\rho_{s_i}-\psi_{s_i}||_1\geq \tr[H_{s_i}(\rho_{s_i}-\psi_{s_i})]\geq \frac{\Delta \eps^2}{4 m}.
\end{align*}
\end{proof}

We note that a pioneering work \cite{Cramer_2010} and its experimental implementation \cite{Lanyon_2017} studied the MPS tomography by focusing on the frustration-free local Hamiltonian and
obtained a similar bound on the unique ground state.  

 A direct application of this result is the following corollary. 
\begin{corollary}
It is sufficient to perform tomography of all the $k$-local reduced density matrices with precision $\frac{\Delta \eps^2}{4 m}$ for trace distance to determine the unique ground state of some $k$-local Hamiltonian up to precision $\eps$ for trace distance. 
\end{corollary}
We employ a specific  overlapping tomography protocol in Section 3, which uses $\mathcal{O}\Big(\frac{10^k \log m}{\delta^2}\Big)$ samples for the tomography of $m$ different $k$-qubit reduced density matrices accurate up to a trace distance parameter $\delta$. Here, we only consider a successful probability greater than  a constant greater than $1/2$, says $2/3$.

From this, we can obtain different sample complexities based on different settings of our knowledge about the connectivity of the local Hamiltonian by choosing $\delta=\frac{\Delta \eps^2}{4m}$. 

\begin{enumerate}
    \item If we do not assume any knowledge of the connectivity of the local Hamiltonian, but only being $k$ local and a gap, we only know that $m\le {n\choose k}$ and hence $\mathcal{O}\Big(\frac{10^k {n \choose k}^2 \log{n \choose k}}{\Delta^2\eps^4}\Big)$ samples suffice for the unique ground state tomography of a $k$-local Hamiltonian $H$ with gap $\Delta>0$ and error parameter $\eps$. 
    \item If we know $m$ but not the interaction graph $G=\{s_1,\cdots,s_m\}$, we must perform tomography on every $k$-local reduced density matrices with error 
parameter $\delta$. Therefore, $\mathcal{O}\Big(\frac{10^k m^2 \log{n \choose k}}{\Delta^2\eps^4}\Big)$ samples suffice.
\item If we know the interaction graph $G=\{s_1,\cdots,s_m\}$, then $\mathcal{O}\Big(\frac{10^k m^2 \log m}{\Delta^2\eps^4}\Big)$ samples suffice.

\end{enumerate}
It may also happen that by merging several $k$-local terms, the resultant $(k+a)$-local terms may have a fewer number, yielding some improvement, where $(k+a)$ is the locality of the merged terms.

\section{Learning the output state of a shallow quantum circuit}
In this section, we study the output state of a quantum circuit and its corresponding parent Hamiltonian.

Before presenting the result, we define $\gamma_2(D)$ for the square lattice.
\begin{definition}
 $\gamma_2(D)$ denotes the largest cardinality of a set of points obtained from a single point set $S_0=\{(0,0)\}$ of the square lattice in $D$ steps, where at each step $j\leq D$, we could get an $S_i$ by adding at most one neighbor point, if it is not in $S_{i-1}$, for each $p\in S_{i-1}$. 
\end{definition}
Through a simple counting argument, one can find the sequence of $\gamma_2(D)$ is $$2,4,8,16,30,\cdots$$ and $$\gamma_2(D)\leq (D+1)^2+D^2\approx 2D^2.$$
This counting does not consider that no two gates can act on the same point in the same layer; hence, it overestimates $\gamma_2(D)$. The following observation builds the connection between the circuit output and the ground state of a local Hamiltonian.
\begin{lemma}\label{circuit-to-Hamiltonian}
The output state $|\psi_{{D}}\rangle$ of $n$ qubit quantum circuit with depth $D\geq 1$ is the unique ground state of a $k$-local frustration-free Hamiltonian  $H=\sum_i H_{s_i}$ with $0\leq H_{s_i}\leq I_{s_i}$ with disjoint $s_i$s, $0\leq H_{s_i}=H_{s_i}^2\leq I_{s_i}$ and a gap at least $1$.  Moreover, $k=2^D$, if the gates are not geometrically local; $k=2D$ for a chain; $k=\gamma_2(D)$ for the square lattice. 
\end{lemma}

\begin{proof}
Let us denote the circuit as\,  $\mathcal{U}\equiv U_{D} \cdots U_1$ and its action on the initial $n$-qubit product state $\ket{0}^{\otimes n}$ gives $ |\psi_{{D}}\rangle=\mathcal{U}\ket{0}^{\otimes n}$, with $U_\alpha$ being the $\alpha$-th layer of unitaries that are composed of non-overlapping two-qubit gates. 

The initial state $\ket{0}^{\otimes n}$ is the unique ground state of 
\begin{align*}
H_0=\sum_i H_{0,i},
\end{align*}
with $H_{0,i}=\op{1}{1}_i$, where the subscript $i$ indicates the qubit site. 
The frustration-free Hamiltonian is $1$-local, including $n$ terms, and has a gap of $1$.

We define the Hamiltonian as follows:
\begin{align*}
 {H_{{D}}}=\mathcal{U} H_0 \mathcal{U}^{\dag}=U_{D} \cdots U_1 H_0 U_{1}^{\dag} \cdots U_D^{\dag}=\sum_i U_{D} \cdots U_1 H_{0,i} U_{1}^{\dag} \cdots U_D^{\dag}.
\end{align*}
$ {H_{{D}}}$ obviously shares the same spectrum as $H_0$, but its locality will be larger. Therefore, $ {H_{{D}}}$ is a frustration-free Hamiltonian with the spectral gap of 1 and the ground state $ |\psi_{{D}}\rangle$.

We observe that each term $H_{D,i}:=U_{D} \cdots U_1 H_{0,i} U_{1}^{\dag} \cdots U_D^{\dag}\geq 0$ is nontrivial on at most $k$ qubits, where
\begin{enumerate}
    \item $k\leq 2^D$, if the gates are not geometrically local;
    \item $k\leq \gamma_2(D)$ on the square  lattice;
    \item $k\leq 2D$ on a 1-dimensional chain.
\end{enumerate}
The above bound on $k$ can be seen from a light cone argument.

If different local terms are nontrivial on a different set of qubits, $ {H_{{D}}}$ is already the local Hamiltonian
with the spectral gap of 1 and the ground state $ |\psi_{{D}}\rangle$. 
Otherwise, there exist some local terms in $\sum_{j=1}^n H_{D,j}:=U_{D} \cdots U_1 H_{0,j} U_{1}^{\dag} \cdots U_D^{\dag}\geq 0$ which are nontrivial on the same set of qubits. In other words, the interaction graph (i.e., generally a hypergraph) of $H_D$ may contain less than $n$ elements. We divide $\{1,\cdots,n\}=\cup_{j} S_j$ such that $i,l\in S_j$ iff $H_{D,i}$ and $H_{D,l}$ are nontrivial on the same set of qubits, and
 define the revised initial Hamiltonian as follows,
\begin{align*}
\widetilde{H}_0=\sum_j \widetilde{H}_{0,j},
\end{align*}
where $\widetilde{H}_{0,j}=I_{S_j}-\op{0\cdots 0}{0\cdots 0}_{S_j}$   with $0$ on qubits in $S_j$.

Further, we let
\begin{align*}
\widetilde {H}_{{D}}=\mathcal{U} \widetilde{H}_0 \mathcal{U}^{\dag}=U_{D} \cdots U_1 \widetilde{H}_0 U_{1}^{\dag} \cdots U_D^{\dag}=\sum_j U_{D} \cdots U_1 \widetilde{H}_{0,i} U_{1}^{\dag} \cdots U_D^{\dag}.
\end{align*}
$\widetilde {H}_{{D}}$ shares the same spectrum as $\widetilde{H}_0$. Therefore, $\widetilde {H}_{{D}}$ is a frustration-free Hamiltonian with a spectral gap of $1$ and the ground state $ |\psi_{{D}}\rangle$ defined above. 

One can verify that different local terms $\widetilde{H}_{D,j}:=U_{D} \cdots U_1 \widetilde{H}_{0,j} U_{1}^{\dag} \cdots U_D^{\dag}\geq 0$ are nontrivial on a different set of qubits.  Furthermore, for $i\in S_j$, $\widetilde{H}_{D,j}$ and ${H_{D,i}}$ apply nontrivially on the same set of qubits. In other words, the locality $k$ of $\widetilde{H}_{D}$ is that same as that of  ${H_{{D}}}$, and specifically, we have
\begin{enumerate}
    \item $k\leq 2^D$, if the gates are not geometrically local;
    \item $k\leq 2D$ on a 1-dimensional chain;
        \item $k\leq \gamma_2(D)$ on the square lattice.
\end{enumerate}
\end{proof}

According to Lemma \ref{ground} by setting $\Delta=1$, we have\begin{theorem}\label{output}
$\ket{\psi}$ has circuit complexity at most $D$. For any state $\rho$, one of the following conditions must be satisfied:
\begin{enumerate}
\item $||\psi-\rho||_1<\eps$;
\item  $||\psi_s-\rho_s||_1>\frac{ \eps^2}{4 n}$ for some $s\subseteq \{0,\cdots, n-1\}$ with $|s|=k$,
\end{enumerate}
where $\psi=\op{\psi}{\psi}$ and we list several scenarios of $k$: 
\begin{enumerate}
    \item $k= 2^D$, if the gates are not geometrically local;
        \item $k= 2D$ on a 1-dimensional chain;
    \item $k= \gamma_2(D)$ on the square lattice.
\end{enumerate}
\end{theorem}

\subsection{Tomography}
According to Theorem \ref{output}, we can accomplish the tomography of the quantum circuit outcome as long as we know that the circuit is of depth at most $D$.
\begin{theorem}
To accomplish the quantum state tomography for depth-$D$ circuit output with precision $\eps$,
\begin{enumerate}
\item $\mathcal{O}\Big(\frac{n^2\cdot10^{k}\log{n \choose k}}{\eps^4}\Big)$ suffice, if we do not know the circuit structure;

\item $\mathcal{O}\Big(\frac{n^2\cdot 10^{k} \log{n}}{\eps^4}\Big)$ suffice, if we know the circuit structure.
\end{enumerate}
Note that, same as above, $k= 2^D$, if the gates are not geometrically local; $k= \gamma_2(D)$ on the square lattice; $k= 2D$ on a 1-dimensional chain.
\end{theorem}
\begin{proof}
If we do not know the circuit structure, it is sufficient to do tomography on all the $k$-local reduced density matrices with precision $\frac{\eps^2}{4 n}$.

If we know the circuit structure, we can compute the $k$-local Hamiltonian, which has at most $n$ terms and at least $1$ spectral gap. 
It is sufficient to perform the tomography of these corresponding $n$ terms of the $k$-local reduced density matrices with precision $\frac{\eps^2}{4 n}$.
\end{proof}

The sample complexity is $poly(n,\frac{1}{\eps})$ for circuits with depth $D$ and, similarly, we list several scenarios:
\begin{enumerate}
\item $D\leq \log\log n+O(1)$ when gates are not geometrically local;
\item $D\leq O(\sqrt{\log n})$ on the square lattice;
\item $D\leq O({\log n})$ on a 1-dimensional chain.
\end{enumerate}

\subsection{Testing the circuit complexity of states}
Using the above results, one can test the circuit complexity of an unknown state. This is manifested in the following theorem.
\begin{theorem}\label{complexity}
For an unknown quantum state $\rho$, and a given $D$, 
 $\mathcal{O}\Big(\frac{n^2\cdot10^{k}\log{n \choose k}}{\eps^4}\Big)$ samples suffice to distinguish between the two cases:
 \begin{enumerate}
 \item $||\psi-\rho||<\frac{\eps^2}{12n}$ for some quantum state $\ket{\psi}$ with circuit complexity $\leq D$;
 \item $||\psi-\rho||>\eps$ for any quantum state $\ket{\psi}$ with circuit complexity $\leq D$;
 \end{enumerate}
$k= 2^D$ in the above sample complexity if the gates are not geometrically local; $k= \gamma_2(D)$ on the square lattice; $k= 2D$ on a 1-dimensional chain.
\end{theorem}

\begin{proof}
We prove the correctness of Algorithm \ref{alg:cap}:

\begin{algorithm}
\KwIn{$\rho$ and $D$\;}
\KwOut{``Yes'' if $||\psi-\rho||<\frac{\eps^2}{12n}$ for some quantum state $\ket{\psi}$ with circuit complexity $\leq D$\;
\ \ \ \ \ \ \ \ \ \ \ \ \ \ \ \ \ ``No'' if $||\psi-\rho||>\eps$ for any quantum state $\ket{\psi}$ with complexity $\leq D$;}

Do the overlapping tomography up to precision $\frac{\eps^2}{12n}$ and obtain $\tilde{\rho_{s}}$ for each $|s|=k$\;

Compute a quantum state $\ket{\psi}$ with circuit complexity $\leq D$ such that $||\psi_s-\tilde{\rho_{s}}||_1<\frac{\eps^2}{6n}$ for every $|s|=k$ \; 

\If {such $\ket{\psi}$ does not exist}
{
Return ``No''\;
}
\Else
{Return ``Yes''.}
\caption{Testing circuit complexity}\label{alg:cap}
\end{algorithm}

If $||\psi-\rho||<\frac{\eps^2}{12n}$ for some quantum state $\ket{\psi}$ with circuit complexity $\leq D$, then
\begin{align*}
||\psi_s-\tilde{\rho_{s}}||\leq ||\psi_s-\rho_s||+||\tilde{\rho_{s}}-\rho_s||<||\psi-\rho||+||\tilde{\rho_{s}}-\rho_s||<\frac{\eps^2}{6n}.
\end{align*}
Step 2  shall find a $\ket{\psi}$, and the algorithm will return ``Yes''.

If $||\psi-\rho||>\eps$ for any quantum state $\ket{\psi}$ with circuit complexity $\leq D$, we need to show that the algorithm shall not find a $\ket{\phi}$ at Step 5. Otherwise, for each $|s|=k$, we have
\begin{align*}
||\phi_s-\rho_{s}||\leq ||\phi_s-\tilde{\rho_s}||+||\tilde{\rho_{s}}-\rho_s||<\frac{\eps^2}{12n}+\frac{\eps^2}{6n}=\frac{\eps^2}{4n}.
\end{align*}
Theorem \ref{output} implies 
\begin{align*}
||\rho-\phi||<\eps,
\end{align*}
which leads to a contradiction!

\end{proof}
It is not hard to see a lower bound of $\Omega(\frac{n}{\eps^2})$ samples is needed, e.g., by considering  depth-1 quantum circuits consisting of one-qubit unitaries and by studying the topography of a tensor product state. The intriguing question is whether the $\frac{1}{\eps^4}$ is necessary using local measurements.

We note that the classical computation in Step 2 is not necessarily easy. The decision version of this problem is a variant of the quantum marginal problem \cite{Klyachko_2006,liu2007consistency}, which focuses on the search for low complexity states.  

It would be interesting to know this problem's precise complexity class. 

\section{Lower bound of the quantum state complexity}
 In this section, we will continue to study the circuit complexity of quantum states but  the perspective of lower bound. 
\begin{definition}
For a quantum state $\ket{\psi}$, its circuit complexity is defined as the minimum depth of quantum circuit $C$ such that
\begin{align*}
\ket{\psi}=C\ket{0}^{\otimes n}.
\end{align*}
\end{definition}

Fact \ref{Hamiltonian} and Lemma \ref{circuit-to-Hamiltonian} imply the following.
\begin{theorem}\label{bound}
If $\ket{\psi}$ is not UDA by its $r$ local reduced density matrices, its circuit complexity is at least: $\log (r+1)$ for non-geometrical circuits, $\lceil \frac{r+1}{2} \rceil $ on $1$-D chain, and $\max\limits_{D:\gamma_2(D)\leq r+1}\lceil D \rceil $ on the square lattice.
\end{theorem}
\begin{proof}
Suppose $\ket{\psi}$ is the output state of a depth-$D$ circuit. According to Lemma \ref{circuit-to-Hamiltonian}, $\ket{\psi}$ is the unique ground state of a $k$-local Hamiltonian, where $k=2^D$ for non-geometrical circuits, $k=\gamma_2(D)$ for the square lattice and $2D$ for $1$-D chain. Fact \ref{Hamiltonian} implies that it is UDA by its $k$-local reduced density matrices. That means
\begin{align*}
k>r.
\end{align*}
This proves our statement.
\end{proof}

It is worth mentioning that the bounds are tight according to the following GHZ example.

\subsection{Examples: GHZ state, long-range entangled and short-range  states}
Take the GHZ state $\ket{\psi}=\frac{1}{\sqrt{2}}(\ket{0}^{\otimes n}+\ket{1}^{\otimes n})$ as an example, it is not UDA by its $n-1$ reduced density matrices. To see this, we observe its $n-1$ reduced density matrices are all 
\begin{align*}
\frac{1}{2}(\otimes_{i=1}^{n-1}\op{0}{0}+\otimes_{i=1}^{n-1}\op{1}{1}).
\end{align*}
The following state has the same $n-1$ reduced density matrices as $\ket{\psi}$
\begin{align*}
\frac{1}{2}(\otimes_{i=1}^{n}\op{0}{0}+\otimes_{i=1}^{n}\op{1}{1}).
\end{align*}
Theorem \ref{bound} implies that the circuit complexity of GHZ is at least
$\lceil \log n\rceil$ for non-geometrical circuits, $\lceil \frac{n}{2} \rceil $ on $1$-D chain, and $\max\limits_{D: \gamma_2(D)\leq n}\lceil D \rceil $ on the square lattice.

All these bounds are tight. The upper bound for non-geometrical circuits follows from the following arguments: In the first layer, a single two-qubit gate can generate a Bell state $\frac{1}{\sqrt{2}}\ket{00}+\ket{11}$; In the second layer, we can use each qubit of the Bell state as a control to generate a $4$-qubit GHZ state; and we continue in this fashion to grow the system size.  In the $\lceil \log n \rceil$-th layer, we can use all the previous qubits of the GHZ state to generate $L$-qubit GHZ state for any $L\leq 2^{\lceil \log n \rceil}$.

We obtain the upper bounds of $\lceil \frac{n}{2} \rceil $ on $1$-D chain and $\max\limits_{D: \gamma_2(D)\leq n}\lceil D \rceil $ on the square lattice similarly. One only needs to have the light cone argument to see that the maximally possible size of a light cone is $2D$ on a $1$-D chain and 
$\gamma_2(D)$ on the square lattice, respectively.

On the other hand, if we allow doing quantum measurement during the quantum circuit, we can generate GHZ state more efficiently: 
In~\cite{verresen2021efficiently}, the authors provide a protocol to generate an $n$-qubit GHZ state. In the first step, they use a depth-2 circuit to create a graph state on a circle of $2n$. In the second step, they perform single-qubit measurements (in the X basis) on half of them. The last action corrects the potential phase change according to the measurement outcomes.

The GHZ state has long-range entanglement. There are other long-range entangled states~\cite{chen2010local} that are topologically ordered, such as the toric code~\cite{kitaev2003fault} or string-nets~\cite{levin2005string}. It is known that they cannot be created from a product state by a finite-depth quantum circuit with geometrically local gates~\cite{chen2010local} but can be created with a linear depth in the system size. However, this does not lead to any specific useful lower bound on the locality of reduced density matrices that enable their unique determination. It is also known that circuits can create them with $\mathcal{O}(\log(N))$-depth, but with long-range gates, in a reverse real-space renormalization procedure~\cite{aguado2008entanglement,konig2009exact}, where $N$ is the total number of qudits. For example,  Ref.~\cite{aguado2008entanglement} considers disentangling the toric code state and gives a scheme in which an operation of  7 (non-local) CNOT depths can reduce the system size by a factor of 4. This gives that the number of depth $D$ to disentangle all qubits is $D=(7/2)\log_2 N$, leading $k\ge 2^D=N^{(7/2)}$, not a useful lower bound. One would expect reduced density matrices to be proportional to the system size (more precisely, of logical code distance)  to determine which degenerate ground state is produced uniquely. 
Although from the circuit complexity perspective, one cannot directly lower bound the range of the reduced density matrices so that a long-range entangled state is UDA, it may be more appropriate from the perspective of topological code structure and its distance~\cite{kitaev2003fault}. In particular, we expect that the code distance provides a lower bound on $k$. 

On the other hand, short-range entangled states can be created from geometrically local gates with a constant depth $D$~\cite{chen2010local} without any symmetry constraint. (Even with symmetry, geometrically non-local gates can disentangle symmetry-protected topologically ordered states with finite-depth circuits~\cite{stephen2022non-local}.) In the worst-case scenario, these states can be uniquely determined by their local reduced density matrices with locality $k\ge 2^D$, assuming 2-local  gates (note geometrically local gates will yield smaller locality $k$). 
Such a property of UDA for short-range entangled states extends to the entire gapped phase (except at the phase boundary), as a finite depth of local gates can connect any  two points inside the same phase. Thus, UDA can be a useful property in short-range gapped phases.

\section{Conclusion and Discussion}
This paper shows that the sample complexity of tomography is low for quantum states with low circuit complexity because learning marginals suffices for state tomography even without knowledge of the circuit structure. Our result aligns with the intuition that the lower the complexity of the quantum state, the fewer samples are needed for learning.
The previous exponential lower bound seems to originate from considering the general quantum states with exponential circuit complexity~\cite{holevo-book,Hayashi_1998}. Our findings thus pave the way for studying the relationship between sample complexity of learning and circuit complexity.  
One exciting question is to complete the picture in the intermediate regime, i.e., to determine the sample complexity of quantum state tomography for quantum states with \textit{polynomial} circuit complexity, using Pauli measurements. 

From a software perspective, data structures are essential for programmers, which enable efficient data storage and retrieval, algorithm design, resource management, and performance optimization.
Our results provide a promising choice as a data structure for quantum computing. The quantum state vector is the primary data structure used in quantum computing, which exhibits the exponential wall in the cost of classical description. Our work shows that the tuple of reduced density matrices is a potential candidate since it is economical and precise in many essential scenarios, including for shallow circuits most relevant to NISQ devices. Additionally, our lower bound describes the bottleneck of quantum gate synthesis and quantum circuit optimization, which will be beneficial to understanding the performance and feasibility of quantum algorithms in comprehensive quantum advantages.

Certification of quantum computation is a timely challenge in quantum technologies because it is essential for developing practical quantum applications, particularly in the NISQ era. Our approach provides a solid theoretical justification and vindication of methods by reduced density matrices. Moreover, potential testing schemes on various NISQ devices suffice to perform only a few local Pauli measurements.

Understanding complex condensed matter systems and facilitating quantum computation relies on the fundamental concept of many-body entanglement. The relationship between the unique ground state of local Hamiltonians and UDA may offer a new perspective to studying many-body quantum phases, as we have seen in our discussion of short-range entangled gapped phases. It is also interesting to explore the relationship between quantum phase 
transitions and the geometry of the reduced density matrices~\cite{Verstraete_2006,PhysRevA.93.012309,PhysRevA.106.012434}, as well as to understand the unique determinism of general tensor network states~\cite{Cirac_2021,https://doi.org/10.48550/arxiv.0707.2260}. Moreover, extending our framework to study long-range entangled, topologically ordered states would be desirable.

An explicit construction from reduced density matrices is generally nontrivial unless it is a ground state of some frustration-free local Hamiltonian. Even in the latter case, it becomes challenging with statistical fluctuations from measurement. For states that MPS or PEPS can approximately describe, 
one possible approach is to compute a parent Hamiltonian~\cite{https://doi.org/10.48550/arxiv.quant-ph/0608197,Cramer_2010}. It is even more challenging when we want to reconstruct a quantum circuit that has the same/close output state. The reasons include the non-uniqueness of local Hamiltonians and the unclear method of transforming local Hamiltonians into quantum circuits.
 
Another potentially exciting direction is to extend our results to Hamiltonian with degenerate ground states and quantum circuits with mid-circuit measurements. The development along these directions may provide a new perspective on topological order. 

We end with the 
optimization problems over complexity $D$ states. Besides the quantum marginal problem in Section 5, the local Hamiltonian problems over complexity $D$ states are also of significant interest. It means that given a $k$-local Hamiltonian $H$, to find the smallest $\tr(H\rho)$ over all $\rho$ with complexity at most $D$. One can also study the decision version. The study of this direction has the potential to enrich the understanding of quantum complexity theory and will have applications in quantum chemistry.


\section{Acknowledgement}
After posting this paper on arXiv, Hsin-Yuan Huang and Daniel Stilck França pointed out that for shallow circuits with known circuit structure, the appendix of~\cite{rouze2021learning} provides an efficient learning algorithm. We thank them for informing us of this. Our paper also solves the open question of learning the unique ground state of local Hamiltonian left in the same paper. We also appreciate Martin Plenio for pointing out the paper \cite{Lanyon_2017} on MPS tomography and short quenches (a natural analog of low-depth quantum circuits) generated state tomography.

We also thank Otfried Guehne, Lorenza Viola, Minbo Gao and Zhengfeng Ji for pointing out a mistake in our argument about UDA and the ground state in our previous version. 
We thank Hsin-Yuan Huang and Matthias Caro for additional and helpful discussions of the sample complexity of quantum states with polynomial circuit complexity using classical shadows. 

T.-C.W. acknowledges the  support of the National Science Foundation under Grant No. PHY 2310614 (in particular, for the part on simulating and learning many-body physics) and by the Materials Science and Engineering Divisions, Office of Basic Energy Sciences of the U.S. Department of Energy under Contract No. DESC0012704 (in particular, for the part on one-dimensional systems).

\begin{thebibliography}{BBMnTR04}

\bibitem[Aar18]{10.1145/3188745.3188802}
Scott Aaronson.
\newblock Shadow tomography of quantum states.
\newblock In {\em Proceedings of the 50th Annual ACM SIGACT Symposium on Theory
  of Computing}, STOC 2018, page 325–338, New York, NY, USA, 2018.
  Association for Computing Machinery.

\bibitem[ABN22]{https://doi.org/10.48550/arxiv.2206.13228}
Anurag Anshu, Nikolas~P. Breuckmann, and Chinmay Nirkhe.
\newblock {NLTS} {H}amiltonians from good quantum codes, 2022.

\bibitem[AV08]{aguado2008entanglement}
Miguel Aguado and Guifr\'e Vidal.
\newblock Entanglement renormalization and topological order.
\newblock {\em Phys. Rev. Lett.}, 100:070404, Feb 2008.

\bibitem[BBMnTR04]{BBMR04}
E.~Bagan, M.~Baig, R.~Mu\~noz Tapia, and A.~Rodriguez.
\newblock Collective versus local measurements in a qubit mixed-state
  estimation.
\newblock {\em Phys. Rev. A}, 69:010304, Jan 2004.

\bibitem[BO21]{10.1145/3406325.3451109}
Costin B\u{a}descu and Ryan O'Donnell.
\newblock Improved quantum data analysis.
\newblock In {\em Proceedings of the 53rd Annual ACM SIGACT Symposium on Theory
  of Computing}, STOC 2021, page 1398–1411, New York, NY, USA, 2021.
  Association for Computing Machinery.

\bibitem[CDJ{\etalchar{+}}13]{PhysRevA.88.012109}
Jianxin Chen, Hillary Dawkins, Zhengfeng Ji, Nathaniel Johnston, David Kribs,
  Frederic Shultz, and Bei Zeng.
\newblock Uniqueness of quantum states compatible with given measurement
  results.
\newblock {\em Phys. Rev. A}, 88:012109, Jul 2013.

\bibitem[CGJ{\etalchar{+}}16]{https://doi.org/10.48550/arxiv.1606.07422}
Jianxin Chen, Cheng Guo, Zhengfeng Ji, Yiu-Tung Poon, Nengkun Yu, Bei Zeng, and
  Jie Zhou.
\newblock Joint product numerical range and geometry of reduced density
  matrices.
\newblock 2016.

\bibitem[CGW10]{chen2010local}
Xie Chen, Zheng-Cheng Gu, and Xiao-Gang Wen.
\newblock Local unitary transformation, long-range quantum entanglement, wave
  function renormalization, and topological order.
\newblock {\em Phys. Rev. B}, 82:155138, Oct 2010.

\bibitem[CHL{\etalchar{+}}22]{https://doi.org/10.48550/arxiv.2206.05265}
Sitan Chen, Brice Huang, Jerry Li, Allen Liu, and Mark Sellke.
\newblock Tight bounds for state tomography with incoherent measurements, 2022.

\bibitem[CJL{\etalchar{+}}16]{PhysRevA.93.012309}
Ji-Yao Chen, Zhengfeng Ji, Zheng-Xin Liu, Yi~Shen, and Bei Zeng.
\newblock Geometry of reduced density matrices for symmetry-protected
  topological phases.
\newblock {\em Phys. Rev. A}, 93:012309, Jan 2016.

\bibitem[CPF{\etalchar{+}}10]{Cramer_2010}
Marcus Cramer, Martin~B. Plenio, Steven~T. Flammia, Rolando Somma, David Gross,
  Stephen~D. Bartlett, Olivier Landon-Cardinal, David Poulin, and Yi-Kai Liu.
\newblock Efficient quantum state tomography.
\newblock {\em Nature Communications}, 1(1), dec 2010.

\bibitem[CPGSV21]{Cirac_2021}
J.~Ignacio Cirac, David P{\'{e} }rez-Garc{\'{\i}}a, Norbert Schuch, and Frank
  Verstraete.
\newblock Matrix product states and projected entangled pair states: Concepts,
  symmetries, theorems.
\newblock {\em Reviews of Modern Physics}, 93(4), dec 2021.

\bibitem[CW20]{Cotler_2020}
Jordan Cotler and Frank Wilczek.
\newblock Quantum overlapping tomography.
\newblock {\em Physical Review Letters}, 124(10), Mar 2020.

\bibitem[EHF19]{evans2019scalable}
Tim~J. Evans, Robin Harper, and Steven~T. Flammia.
\newblock Scalable bayesian hamiltonian learning.
\newblock {\em arXiv:1912.07636}, 2019.

\bibitem[FGLE12a]{FlammiaGrossLiuEtAl2012}
S.~T. Flammia, D.~Gross, Y.~Liu, and J.~Eisert.
\newblock Quantum tomography via compressed sensing: Error bounds, sample
  complexity, and efficient estimators.
\newblock {\em New J. Phys.}, 14:095022, 2012.

\bibitem[FGLE12b]{Flammia_2012}
Steven~T Flammia, David Gross, Yi-Kai Liu, and Jens Eisert.
\newblock Quantum tomography via compressed sensing: error bounds, sample
  complexity and efficient estimators.
\newblock {\em New Journal of Physics}, 14(9):095022, Sep 2012.

\bibitem[FH14]{NLTS}
Michael~H. Freedman and Matthew~B. Hastings.
\newblock Quantum systems on non-$k$-hyperfinite complexes: A generalization of
  classical statistical mechanics on expander graphs.
\newblock 14(1–2)(144–180), 2014.

\bibitem[GK08]{GJK08}
M\u{a}d\u{a}lin Gu\c{t}\u{a} and Jonas Kahn.
\newblock Optimal estimation of qubit states with continuous time measurements.
\newblock {\em Communications in Mathematical Physics}, 277(1):127--160, 2008.

\bibitem[GKKT20]{Gu__2020}
M~Gu{\c{t}}{\u{a}}, J~Kahn, R~Kueng, and J~A Tropp.
\newblock Fast state tomography with optimal error bounds.
\newblock {\em Journal of Physics A: Mathematical and Theoretical},
  53(20):204001, apr 2020.

\bibitem[GLF{\etalchar{+}}10]{compressed}
D.~Gross, Y.~Liu, S.~T. Flammia, S.~Becker, and J.~Eisert.
\newblock Quantum state tomography via compressed sensing.
\newblock {\em Phys. Rev. Lett.}, 105(150401), 2010.

\bibitem[GPM{\etalchar{+}}10]{Gawron_2010}
Piotr Gawron, Zbigniew Pucha{\l}a, Jaros{\l}aw~Adam Miszczak, {\L}ukasz
  Skowronek, and Karol {\.{Z} }yczkowski.
\newblock Restricted numerical range: A versatile tool in the theory of quantum
  information.
\newblock {\em Journal of Mathematical Physics}, 51(10):102204, oct 2010.

\bibitem[Hay98]{Hayashi_1998}
Masahito Hayashi.
\newblock Asymptotic estimation theory for a finite-dimensional pure state
  model.
\newblock {\em Journal of Physics A: Mathematical and General},
  31(20):4633--4655, may 1998.

\bibitem[HFK{\etalchar{+}}22]{Haferkamp_2022}
Jonas Haferkamp, Philippe Faist, Naga B.~T. Kothakonda, Jens Eisert, and
  Nicole~Yunger Halpern.
\newblock Linear growth of quantum circuit complexity.
\newblock {\em Nature Physics}, 18(5):528--532, mar 2022.

\bibitem[HHHH09]{RevModPhys.81.865}
Ryszard Horodecki, Pawe\l{} Horodecki, Micha\l{} Horodecki, and Karol
  Horodecki.
\newblock Quantum entanglement.
\newblock {\em Rev. Mod. Phys.}, 81:865--942, Jun 2009.

\bibitem[HHJ{\etalchar{+}}16]{HHJ+16}
J.~Haah, A.~W. Harrow, Z.~Ji, X.~Wu, , and N.~Yu.
\newblock Sample-optimal tomography of quantum states.
\newblock In {\em Proceedings of the Forty-Seventh Annual ACM on Symposium on
  Theory of Computing}, STOC '16, pages 913--925, 2016.

\bibitem[HKP20]{Huang_2020}
Hsin-Yuan Huang, Richard Kueng, and John Preskill.
\newblock Predicting many properties of a quantum system from very few
  measurements.
\newblock {\em Nature Physics}, Jun 2020.

\bibitem[HKT{\etalchar{+}}22]{Huang_2022}
Hsin-Yuan Huang, Richard Kueng, Giacomo Torlai, Victor~V. Albert, and John
  Preskill.
\newblock Provably efficient machine learning for quantum many-body problems.
\newblock {\em Science}, 377(6613), sep 2022.

\bibitem[Hol73]{Holevo73}
A.~S. Holevo.
\newblock Bounds for the quantity of information transmitted by a quantum
  communication channel.
\newblock {\em Problems of Information Transmission}, 9:177--183, 1973.

\bibitem[Hol82]{holevo-book}
A.S. Holevo.
\newblock {\em Probabilistic and Statistical Aspects of Quantum Theory}.
\newblock North Holland, 1982.

\bibitem[Key06]{Keyl06}
M.~Keyl.
\newblock Quantum state estimation and large deviations.
\newblock {\em Reveiws in Mathematical Physics}, 18(1):19--60, 2006.

\bibitem[Kit03]{kitaev2003fault}
A~Yu Kitaev.
\newblock Fault-tolerant quantum computation by anyons.
\newblock {\em Annals of physics}, 303(1):2--30, 2003.

\bibitem[KJTV19]{Karuvade_2019}
Salini Karuvade, Peter~D. Johnson, Francesco Ticozzi, and Lorenza Viola.
\newblock Uniquely determined pure quantum states need not be unique ground
  states of quasi-local hamiltonians.
\newblock {\em Physical Review A}, 99(6), jun 2019.

\bibitem[Kly06]{Klyachko_2006}
Alexander~A Klyachko.
\newblock Quantum marginal problem and n-representability.
\newblock {\em Journal of Physics: Conference Series}, 36:72--86, apr 2006.

\bibitem[KRT17]{KRT14}
R.~Kueng, H.~Rauhut, and U.~Terstiege.
\newblock Low rank matrix recovery from rank one measurements.
\newblock {\em Applied and Computational Harmonic Analysis}, 42:88--116, 2017.

\bibitem[KRV09]{konig2009exact}
Robert K\"onig, Ben~W. Reichardt, and Guifr\'e Vidal.
\newblock Exact entanglement renormalization for string-net models.
\newblock {\em Phys. Rev. B}, 79:195123, May 2009.

\bibitem[Liu07]{liu2007consistency}
Yi-Kai Liu.
\newblock Consistency of local density matrices is qma-complete, 2007.

\bibitem[LMH{\etalchar{+}}17]{Lanyon_2017}
B.~P. Lanyon, C.~Maier, M.~Holzäpfel, T.~Baumgratz, C.~Hempel, P.~Jurcevic,
  I.~Dhand, A.~S. Buyskikh, A.~J. Daley, M.~Cramer, M.~B. Plenio, R.~Blatt, and
  C.~F. Roos.
\newblock Efficient tomography of a quantum many-body~system.
\newblock {\em Nature Physics}, 13(12):1158--1162, sep 2017.

\bibitem[LPW02]{PhysRevLett.89.207901}
N.~Linden, S.~Popescu, and W.~K. Wootters.
\newblock Almost every pure state of three qubits is completely determined by
  its two-particle reduced density matrices.
\newblock {\em Phys. Rev. Lett.}, 89:207901, Oct 2002.

\bibitem[LW02]{PhysRevLett.89.277906}
N.~Linden and W.~K. Wootters.
\newblock The parts determine the whole in a generic pure quantum state.
\newblock {\em Phys. Rev. Lett.}, 89:277906, Dec 2002.

\bibitem[LW05]{levin2005string}
Michael~A Levin and Xiao-Gang Wen.
\newblock String-net condensation: A physical mechanism for topological phases.
\newblock {\em Physical Review B}, 71(4):045110, 2005.

\bibitem[McD89]{mcdiarmid_1989}
Colin McDiarmid.
\newblock {\em On the method of bounded differences}, page 148–188.
\newblock London Mathematical Society Lecture Note Series. Cambridge University
  Press, 1989.

\bibitem[Mou16]{Mousavi2016}
Nima Mousavi.
\newblock How tight is chernoff bound?
\newblock {\em
  \url{https://ece.uwaterloo.ca/~nmousavi/Papers/Chernoff-Tightness.pdf}},
  2016.

\bibitem[NC11]{Nielsen:2011:QCQ:1972505}
M.~A. Nielsen and I.~L. Chuang.
\newblock {\em Quantum Computation and Quantum Information: 10th Anniversary
  Edition}.
\newblock Cambridge University Press, 10th edition, 2011.

\bibitem[OW16]{OW16}
R.~O'Donnell and J.~Wright.
\newblock Efficient quantum tomography.
\newblock In {\em Proceedings of the Forty-Seventh Annual ACM on Symposium on
  Theory of Computing}, STOC '16, pages 899--912, 2016.

\bibitem[OW17]{OW17}
R.~O'Donnell and J.~Wright.
\newblock Efficient quantum tomography ii.
\newblock In {\em Proceedings of the Forty-Seventh Annual ACM on Symposium on
  Theory of Computing}, STOC '17, pages 962--974, 2017.

\bibitem[Par79]{PhysRevLett.43.1754}
G.~Parisi.
\newblock Infinite number of order parameters for spin-glasses.
\newblock {\em Phys. Rev. Lett.}, 43:1754--1756, Dec 1979.

\bibitem[Par83]{PhysRevLett.50.1946}
Giorgio Parisi.
\newblock Order parameter for spin-glasses.
\newblock {\em Phys. Rev. Lett.}, 50:1946--1948, Jun 1983.

\bibitem[Par99]{parisi1999complex}
Giorgio Parisi.
\newblock Complex systems: a physicist's viewpoint.
\newblock {\em Physica A: Statistical Mechanics and its Applications},
  263(1):557--564, 1999.
\newblock Proceedings of the 20th IUPAP International Conference on Statistical
  Physics.

\bibitem[PGM{\etalchar{+}}11]{Pucha_a_2011}
Zbigniew Pucha{\l}a, Piotr Gawron, Jaros{\l}aw~Adam Miszczak, {\L}ukasz
  Skowronek, Man-Duen Choi, and Karol {\.{Z}}yczkowski.
\newblock Product numerical range in a space with tensor product structure.
\newblock {\em Linear Algebra and its Applications}, 434(1):327--342, jan 2011.

\bibitem[PGVCW07]{https://doi.org/10.48550/arxiv.0707.2260}
David Perez-Garcia, Frank Verstraete, J.~Ignacio Cirac, and Michael~M. Wolf.
\newblock {PEPS} as unique ground states of local hamiltonians.
\newblock 2007.

\bibitem[PGVWC07]{https://doi.org/10.48550/arxiv.quant-ph/0608197}
D.~Perez-Garcia, F.~Verstraete, M.~M. Wolf, and J.~I. Cirac.
\newblock Matrix product state representations.
\newblock 7(401), 2007.

\bibitem[RF21]{rouze2021learning}
Cambyse Rouzé and Daniel~Stilck França.
\newblock Learning quantum many-body systems from a few copies, 2021.

\bibitem[SDLN22]{stephen2022non-local}
David~T. Stephen, Arpit Dua, Ali Lavasani, and Rahul Nandkishore.
\newblock Non-local finite-depth circuits for constructing spt states and
  quantum cellular automata, 2022.

\bibitem[SDV06]{Shi_2006}
Y.-Y. Shi, L.-M. Duan, and G.~Vidal.
\newblock Classical simulation of quantum many-body systems with a tree tensor
  network.
\newblock {\em Physical Review A}, 74(2), aug 2006.

\bibitem[TD04]{https://doi.org/10.48550/arxiv.quant-ph/0205133}
Barbara~M. Terhal and David~P. DiVincenzo.
\newblock Adaptive quantum computation, constant depth quantum circuits and
  arthur-merlin games.
\newblock 4(2), 2004.

\bibitem[vACGN23]{doi:10.1137/1.9781611977554.ch47}
Joran van Apeldoorn, Arjan Cornelissen, András Gilyén, and Giacomo Nannicini.
\newblock {\em Quantum tomography using state-preparation unitaries}, pages
  1265--1318.
\newblock ACM, 2023.

\bibitem[VC06]{Verstraete_2006}
F.~Verstraete and J.~I. Cirac.
\newblock Matrix product states represent ground states faithfully.
\newblock {\em Physical Review B}, 73(9), mar 2006.

\bibitem[VTV21]{verresen2021efficiently}
Ruben Verresen, Nathanan Tantivasadakarn, and Ashvin Vishwanath.
\newblock Efficiently preparing schrödinger's cat, fractons and non-abelian
  topological order in quantum devices, 2021.

\bibitem[WHG18]{Wyderka_2018}
Nikolai Wyderka, Felix Huber, and Otfried Gühne.
\newblock Constraints on correlations in multiqubit systems.
\newblock {\em Physical Review A}, 97(6), jun 2018.

\bibitem[WSSM22]{PhysRevA.106.012434}
Samuel Warren, LeeAnn~M. Sager-Smith, and David~A. Mazziotti.
\newblock Quantum simulation of quantum phase transitions using the convex
  geometry of reduced density matrices.
\newblock {\em Phys. Rev. A}, 106:012434, Jul 2022.

\bibitem[XLK{\etalchar{+}}17]{Xin_2017}
Tao Xin, Dawei Lu, Joel Klassen, Nengkun Yu, Zhengfeng Ji, Jianxin Chen, Xian
  Ma, Guilu Long, Bei Zeng, and Raymond Laflamme.
\newblock Quantum state tomography via reduced density matrices.
\newblock {\em Physical Review Letters}, 118(2), jan 2017.

\end{thebibliography}
\newcommand{\etalchar}[1]{$^{#1}$}

\end{document}